\documentclass[acmsmall]{acmart}

\usepackage{balance}
\usepackage{amsthm}
\usepackage{amsmath}
\usepackage[ruled,vlined]{algorithm2e}
\usepackage{bm}
\usepackage{dsfont}
\usepackage{booktabs}
\usepackage{enumitem}

\DeclareMathOperator{\EP}{\mathbb{P}}
\DeclareMathOperator{\EX}{\mathbb{E}}
\DeclareMathOperator{\II}{\mathds{1}}
\DeclareMathOperator{\argmax}{argmax}
\DeclareMathOperator{\kmax}{K_{\max}}

\newcommand\EU{\mathit{EU}}
\newcommand\BR{\mathit{BR}}
\newcommand\NR{\mathit{NR}}
\newcommand\NB{\mathit{NB}}
\newcommand\given[1][]{\:#1\vert\:}
\newcommand{\name}{IRS\xspace}

\newcommand{\reftab}[1]{Table \ref{#1}}

\newcommand{\refsec}[1]{\S\ref{#1}}
\newcommand{\refsecs}[2]{\S\ref{#1}--\S\ref{#2}}
\newcommand{\refapp}[1]{Appendix \ref{#1}}

\newcommand{\refequ}[1]{Equation (\ref{#1})}

\newcommand{\refalg}[1]{Algorithm (\ref{#1})}
\newcommand{\refthm}[1]{Theorem (\ref{#1})}
\newcommand{\reflemm}[1]{Lemma (\ref{#1})}

\newtheorem{theorem}{Theorem}%
\newtheorem{definition}[theorem]{Definition}%
\newtheorem{lemma}[theorem]{Lemma}%
\newtheorem{proposition}[theorem]{Proposition}%
\setcopyright{none}
\renewcommand\footnotetextcopyrightpermission[1]{}

\title[\name]{\name: An Incentive-compatible Reward Scheme for Algorand}

\author{Maizi Liao}
\affiliation{
  \institution{University of Waterloo}
  \city{Waterloo}
  \country{Canada}}
\email{m7liao@uwaterloo.ca}

\author{Wojciech Golab}
\affiliation{
  \institution{University of Waterloo}
  \city{Waterloo}
  \country{Canada}}
\email{wgolab@uwaterloo.ca}

\author{Seyed Majid Zahedi}
\affiliation{
  \institution{University of Waterloo}
  \city{Waterloo}
  \country{Canada}}
\email{smzahedi@uwaterloo.ca}

\begin{abstract}
Founded in 2017, Algorand is one of the world's first carbon-negative, public blockchains inspired by proof of stake.
Algorand uses a Byzantine agreement protocol to add new blocks to the blockchain.
The protocol can tolerate malicious users as long as a supermajority of the stake is controlled by non-malicious users.
The protocol achieves about 100x more throughput compared to Bitcoin and can be easily scaled to millions of nodes.
Despite its impressive features, Algorand lacks a reward-distribution scheme that can effectively incentivize nodes to participate in the protocol.
In this work, we study the incentive issue in Algorand through the lens of game theory.
We model the Algorand protocol as a Bayesian game and propose a novel reward scheme to address the incentive issue in Algorand.
We derive necessary conditions to ensure that participation in the protocol is a Bayesian Nash equilibrium under our proposed reward scheme even in the presence of a malicious adversary.
We also present quantitative analysis of our proposed reward scheme by applying it to two real-world deployment scenarios.
We estimate the costs of running an Algorand node and simulate the protocol to measure the overheads in terms of computation, storage, and networking.
\end{abstract}

\begin{document}

\pagestyle{fancy}
\fancyhead{}
\fancyfoot{}

\maketitle 
\thispagestyle{empty}

\section{Introduction}

The concept of blockchain is first popularized by Bitcoin \cite{Nakamoto2008} as a tamper-resistant distributed transaction ledger.
Blockchains could be classified into two categories: permissioned and permissionless.
Permissioned blockchains, also known as private blockchains, implement an access-control mechanism to restrict unauthorized users from accessing the ledger \cite{CL1999,Lamport2001,OO2014}.
Examples include HyperLedger Fabric \cite{ABB+2018} and Libra (now called Diem) \cite{Diem}.
In contrast, permissionless blockchains do not impose any access restrictions \cite{helliar2020, liu2019}.
Examples include Bitcoin \cite{Nakamoto2008} and Ethereum \cite{wood2014ethereum}.

When obtaining permission is not required, the system could become prone to Sybil attacks%
\footnote{In a Sybil attack, the attacker creates a large number of pseudonymous identities to gain disproportionate control and/or influence over the system.}.
To mitigate the Sybil attack threat, permissionless consensus protocols often use additional mechanisms.
For example, Bitcoin uses proof of work (PoW), which requires nodes to solve a computationally intensive puzzle.
Winners earn the right to add blocks to the blockchain and collect rewards for their computational effort.
PoW suffers from high energy and computational costs \cite{BitcoinEnergy}.
Proof of stake (PoS) has been proposed to mitigate these costs \cite{BitcoinWiki-PoS, KN2012}.
In most PoS consensus protocols, nodes stake their cryptocurrency assets to gain rights to add blocks and earn rewards.

Inspired by proof of stake, Algorand is one of the world's first carbon negative blockchain protocols \cite{carbonNegative}.
The Algorand blockchain runs a randomized, committee-based consensus protocol \cite{GHM+2017, CM2019}.
The core of the protocol is a Byzantine agreement protocol that allows nodes to reach consensus on a new block in the presence of Byzantine faults%
\footnote{Byzantine faults cause a node to inconsistently appear both failed and functioning to others.
The term ``Byzantine'' is taken from the ``Byzantine general problem'' \cite{LSP1982}.}.
Nodes are selected randomly to participate in the Byzantine agreement protocol as committee members.
The original reward scheme of Algorand%
\footnote{Algorand is moving from its original reward scheme to a new reward scheme called the Governance Rewards \cite{governance-rewards}.
Under the new reward scheme, only agents who commit to participate in the governance of the Algorand ecosystem will be rewarded.
Agents have to prove their commitment by locking their cryptocurrency assets for a potentially long term.
This new scheme is in line with proof-of-stake protocols.}
rewards all nodes proportionally to their account balance.
Although simple, this reward scheme suffers from the free-rider problem: nodes have no incentive to participate in the protocol as doing so imposes computational and communication costs.
This can be seen by tracking the number of nodes that actively participate in Algorand.
According to \cite{AlgoMetric}, in May 2022, only about 1.6 billion units of Algorand's cryptocurrency were registered to participate while a total of about 8.4 billion units of Algorand's cryptocurrency were available.
Moreover, while there were more than 1.7 million active accounts, only 361 unique accounts had recently participated in Algorand's consensus protocol.
Since the safety and liveness of Algorand depend mainly on high participation of nodes, the lack of participation poses a serious threat to Algorand by making it prone to attacks from malicious participants.

There exist other proof of stake based consensus protocols which have different properties comparing to Algorand.
Avalanche's consensus protocol \cite{AVA} is inspired by gossip protocols and does not need a committee.
Tendermint \cite{BKM2019} is committee-based consensus protocol but the membership of the committee is deterministic and public.
The consensus protocols of Ethereum 2, Polkadot and Cardano are similar. They use proof of stake to randomly select block proposers and then use a chain selection rule to resolve forks \cite{Gasper, CasperFFG, BRL+, GRANDPA, GHOST, Ouroboros, OuroborosPraos}.

The incentive problem in Bitcoin-like blockchains has been studied extensively in recent years \cite{ES2014, SBB+2016, PS2017, LS2019, CHL2021}.
However, the results do not apply to Algorand due to its unique consensus protocol.
In this paper, we propose \name, an incentive-compatible reward scheme for Algorand.
Under \name, nodes' participation is monitored through committee votes.
These votes serve as evidence of the committee members' participation in the protocol.
In addition, \name requires committee members to include with their vote the identities of the nodes from which they have received a valid message.
This facilitates the monitoring of the collaboration of nodes that are not chosen as committee members.
We model Algorand as a Bayesian game and study nodes' strategies under our proposed reward scheme.
Our analysis considers an adversary that can corrupt nodes in a probabilistic manner.
We show that if certain conditions are met, all nodes are incentivized to participate in the protocol regardless of being selected as committee members.

In summary, we make the following contributions. In \refsec{sec:cost-model}, we present a detailed cost model for nodes' participation in Algorand. In \refsec{sec:algorand-game}, we model the Algorand protocol as a Bayesian game. In \refsec{sec:analysis}, we propose \name, a novel incentive-compatible reward scheme to address the free-rider problem. In \refsec{sec:analysis}, we study equilibrium strategies under \name and derive necessary conditions to ensure nodes participation. In \refsec{sec:impl}, we present detailed implementation requirements for real-world deployment of \name. Finally, in \refsec{sec:exp}, we quantitatively analyze our proposed reward scheme and simulate the protocol to measure its overheads in terms of computation, storage, and networking.

\section{Algorand Protocol}
\label{sec:algorand-prot}

The Algorand protocol maintains a permissionless blockchain.
In Algorand, adding a new block to the blockchain requires multiple steps.
\refalg{alg:algorand} provides high-level pseudocode of the Algorand protocol (please see \refapp{sec:algorand} for an in-depth overview of the protocol).
At each step, all nodes wait for the messages from the previous step for a fixed amount of time.
Each node then validates and propagates received messages to its neighbors.
The protocol can terminate at specific steps (i.e., $k > 4$ where $k \not\equiv 1 \pmod 3$) if a termination condition is met (i.e., enough votes are received or the final step, $\kmax$, is reached).
At each non-terminal step, a random committee of self-selected nodes is formed.
Committee members generate and propagate a message according to the protocol.
The message is a block proposal if $k = 1$, and its a vote on a proposed block if $k > 1$.

\noindent
\textbf{Cryptographic sortition.}
In Algorand, nodes are assumed to have access to a unique signature scheme (e.g., \cite{micali1999verifiable}).
As shown in \refalg{alg:algorand}, at step $k$, given a publicly known random seed, $Q$, each node privately computes a hashlen-bit-long random string $x_{k} = H(\text{SIG}(k, Q))$ by digitally signing $(k, Q)$ and then hashing it using a random oracle $H$ \cite{goldreich2007foundations}.
The string $x_{k}$ is interpreted as a binary expansion of a number between 0 and 1, denoted by $.x_{k} = x_k / 2^{\text{hashlen}}$.
If this number is less than a known threshold, $p_c$, then the node is a member of the committee at step $k$.
The threshold is set such that the expected size of the committee is $\tau$ (i.e., $p_c = \tau / n$, where $n$ is the number of nodes in the system).
$.x_{k}$ is also used to represent the priority of the node; the smaller $.x_{k}$ is the higher the priority of the node will be.
For nodes in the committee, $\sigma_{k} = \text{SIG}(k, Q)$ is the committee credential.
Committee members propagate their credential alongside their generated message.

\begin{algorithm}[!t]
\SetKwFunction{Sort}{Sortition}%
\SetKwProg{Proc}{procedure}{ begin}{end}%
\For{$k = 1, \dots, \kmax$}{
    \lIf{$k > 1$}{
        Validate and gossip step-($k$ - 1) messages%
    }
    
    \lIf{$k > 4$ and $k \not\equiv 1 \pmod 3$}{
        Exit if termination condition for step $k$ is met%
    }
    
    $(x_{k}, \sigma_{k})$ $\gets$ \Sort{$k$}\;
    \If{$.x_{k}$ $\le$ $p_c$}{
        message $\gets$ Generate a message\;
        Propagate message and $\sigma_{k}$\;
    }
}
\Proc{\Sort{$k$}}{
    $\sigma$ $\gets$ $\text{SIG}(k, Q)$\;
    $x$ $\gets$ $H(\sigma)$\;
    \Return ($x$, $\sigma$)\;
}
\caption{High-level pseudocode for Algorand}
\label{alg:algorand}
\end{algorithm}

\noindent
\textbf{Gossip network and protocol.}
In Algorand, each node is provided with an \emph{address-book} file containing the IP address and the port number of other nodes.
Nodes form a gossip network by selecting a subset of $n_{rp}$ random peers to gossip messages to.
The parameter $n_{rp}$ depends on the number of nodes, and it is set such that the gossip network is strongly connected.
Messages are disseminated on the gossip network using a gossip protocol.
The message dissemination is initiated by committee members at each step.
Each committee member propagates their generated message to their randomly selected peers.
Those peers then forward the message to their own peers.
And this process continues until the message is received by all the nodes in the network.
To avoid forwarding loops, nodes do not propagate the same message twice.
\section{Preliminaries}

\subsection{Adversary Model}
The committee is guaranteed to reach Byzantine agreement \cite{pease1980reaching, dolev1983authenticated, fischer1983consensus} in the presence of an adversary that can corrupt nodes and control their actions.
The Algorand protocol is resilient to such adversary as long as it cannot corrupt more than 1/3 of the nodes.
This is achieved by setting the expected size of the committee, $\tau$, such that, with high probability, at least 2/3 of the committee members are non-Byzantine nodes.
In this paper, however, we consider a slightly different adversary model.
In particular, we assume that the adversary corrupts each node with a fixed probability $p_b < 1/3$.
Under our probabilistic adversary model, it is possible for the adversary to corrupt more than 1/3 of nodes ex post.
However, we show in \refsec{sec:analysis} that for large systems, under our adversary model, non-Byzantine nodes still constitute more than 2/3 of the committee with high probability.
Consequently, Algorand protocol is guaranteed to reach Byzantine agreement with high probability.

\noindent
\textbf{Non-Byzantine supermajority.}
Since Algorand is a permissionless blockchain, the adversary can easily introduce as many new nodes as it wishes.
Therefore, instead of assuming that the system has at least a 2/3 majority of non-Byzantine nodes, it is often more meaningful to assume that at least 2/3 of the cryptocurrency assets are controlled by non-Byzantine nodes.
In other words, instead of assuming that the adversary can corrupt up to 1/3 of the nodes, it is often assumed that the adversary can control up to 1/3 of the assets in the blockchain.
Algorand achieves this by assigning \emph{sub-nodes} to each node in proportion to the balance of its account.
The cryptographic sortition algorithm then randomly selects each sub-node as a committee member.
In this paper, we present our analysis under the simpler assumption that each node has a single sub-node.
We then show how to modify the Algorand protocol and our analysis to consider the more realistic assumption that each node controls multiple sub-nodes.

\subsection{Network Model}
In this paper, we assume that the gossip network is strongly synchronous.
This is a widely adopted network assumption \cite{GHM+2017,BKM2019,ABP+2020,ABP+2020a} which states that all messages propagated initially by non-Byzantine nodes are received by all other non-Byzantine nodes within a known time period.
We further assume that the network remains strongly synchronous if a majority of nodes run the gossip protocol.
This means that for large systems, the adversary cannot launch an Eclipse attack \cite{heilman2015eclipse} with high probability.

\subsection{Cost Model}
\label{sec:cost-model}
Nodes running the Algorand protocol incur processing and communication costs at each step of the protocol.
These costs are measurable in quantitative terms (e.g., energy consumption)  and can be expressed in monetary values (e.g., cryptocurrency or Dollar).

We denote the total cost incurred by any node $i$ at step $k$ by $C_i(k)$.
This cost has two components: (a) baseline costs, $C_i^b(k)$, and (b) committee costs, $C^c(k)$.
We model $C_i(k)$ as follows.
\begin{equation}
\label{equ:cost}
C_i(k) = C_i^{b}(k) + C^{c}(k) \times \II(.x_{i,k} < p_c)\footnotemark.
\end{equation}
\footnotetext{$\II(\cdot)$ is the indicator function which returns $1$ if the condition is true, and $0$ otherwise.}
$C_i^{b}(k)$ represents the baseline costs that do not depend on whether node $i$ is selected as a committee member at step $k$.
Examples include costs of running cryptographic sortition and propagating messages.
$C^{c}(k)$ represents the committee costs incurred by a node when it is selected as a committee member at step $k$.
Examples include costs of generating blocks and votes.

\section{The Algorand Game}
\label{sec:algorand-game}
To study nodes' incentives, we model the participation of nodes in the Algorand protocol as a Bayesian game.
We formally define the Algorand game and describe nodes' strategies and utilities.
We then discuss solution concepts for our proposed game.

\subsection{Game Model}
We model the Algorand protocol as a Bayesian game.
A Bayesian game consists of a set of agents.
Each agent has a type and a set of available actions.
Agents do not know their types before the start of the game.
They, however, know a common prior probability distribution over types.
At the beginning of the game, each agent privately observes its own type.
Agents then simultaneously take their actions without knowing each others' types.
Finally, agents receive a real-valued utility (a.k.a. payoff) given their joint types and actions.
A Bayesian game is formally defined as follows.
\begin{definition}[Bayesian Game \cite{SL2009}]
A Bayesian game is represented by a tuple $(N, A, \Theta, p, u)$ where:
\begin{itemize}[noitemsep=1pt,topsep=2pt,leftmargin=*]
    \item $N = \{ 1, \dots, n \}$ is a set of agents;
    \item $A = A_1 \times \dots \times A_n$, where $A_i$ is a set of actions available to agent $i$;
    \item $\Theta = \Theta_1 \times \dots \times \Theta_n$, where $\Theta_i$ is the type space of agent $i$;
    \item $P: \Theta \mapsto [0, 1]$ is a common prior over types; and
    \item $u = (u_1, \dots, u_n)$, where $u_i: A \times \Theta \mapsto \mathbb{R}$ is the utility for agent $i$.
\end{itemize}
\end{definition}

\noindent
\textbf{Agents and actions.}
In our setting, agents represent Algorand nodes.
Agents are assumed to be rational in the sense that they selfishly choose an action to maximize their utility function.
We consider three actions: (i) cooperate, $C$, (ii) defect, $D$, (iii) misbehave, $M$.
A cooperative agent fully runs the Algorand protocol's code and consequently incurs all the processing and communication costs associated with it.
A defective agent does not run any code (e.g., logs off from the system) and incurs no costs.
A misbehaving agent runs a malicious code to sabotage the system.
Note that the malicious code can imitate the behaviour of cooperating or defecting.
Non-Byzantine agents that are not corrupted by the adversary do not misbehave.
They only choose between cooperating and defecting.
They cooperate if and only if their expected rewards exceed their expected costs.
Byzantine agents, however, always misbehave (i.e., they run the adversary's malicious code).

Formally, we denote the action of agent $i$ by $a_i \in A_i = \{ C, D, M \}$.
A vector of actions $a = (a_1, \dots, a_n) \in A$ is called an action profile.
An action profile $a$ can be written as $(a_i, a_{-i})$, where $a_{-i}$ is an action profile without agent $i$'s action%
\footnote{Throughout the paper, we use $-i$ to denote all agents except agent $i$.}.

\noindent
\textbf{Types.}
Type of agent $i$ is defined to be $\theta_i = (\theta_{i,0}, \theta_{i,1}, \dots, \theta_{i,\kmax})$
where $\theta_{i,0} \in \{0, 1\}$ indicates if agent $i$ is corrupted by the adversary (1 if $i$ is Byzantine and zero otherwise), and $\theta_{i,k}$ is a number between 0 and 1 represented in binary by the hash result of the sortition algorithm run by agent $i$ at step $k = 1, \dots, \kmax$ ($\theta_{i,k} = .x_{i,k}$, where $x_{i,k}$ is returned by $\texttt{Sortition}_i(k)$).
The type space of agent $i$ is denoted as $\Theta_i$.
At the beginning of the Algorand game, each agent observers whether it is corrupted by the adversary.
Agents also receive a random seed for the cryptographic sortition algorithm.
Given the random seed, agents can run the sortition algorithm for all steps to know what exactly their type vector is.

\noindent
\textbf{Prior probabilities.}
The probability that agent $i$ is corrupted by the adversary is $\EP(\theta_{i,0} = 1) = p_b$.
For $k = 1, \dots, \kmax$, $\theta_{i,k}$'s are drawn independently from a uniform distribution between 0 and 1.
Therefore, the probability that agent $i$  is  selected as a committee member is $\EP(\theta_{i,k} \le p_c) = p_c$.
The adversary corrupts agents independently.
Agents also are selected as committee members at any given step independently.

\noindent
\textbf{Utility functions.}
For agent $i$, the utility function, $u_i$, maps action profiles, $a=(a_1,\dots, a_n)$, and type vectors, $\theta = (\theta_1,\dots,\theta_n)$, to real-valued payoffs.
If $\theta_{i,0} = 1$, then we assume that $u_i(a, \theta)$ is $-B$ if $a_i \in \{C, D\}$ and $0$ if $a_i = M$ where $B$ is a large real number.
Under this assumption, regardless of their type, Byzantine agents always prefer $M$ to $C$ and $D$.
To model non-Byzantine agents preferences, we assume that if $\theta_{i,0} = 0$, then $u_i(a, \theta) = -B$ if $a_i = M$.
For $a_i \in \{C, D\}$, the utility of non-Byzantine agents is equal to the rewards they receive minus the costs they incur.
We analyze utility functions in more details in \refsec{sec:analysis}.

\subsection{Strategies and Equilibria}
\label{sec:strategies-equilibria}
A strategy defines a description of how a game would be played in every contingency.
In a Bayesian game, a strategy prescribes a distribution over actions for every type that an agent could have.
Let $\Delta(A_i)$ be the set of all probability distributions over $A_i$.
For agent $i$, a strategy $s_i: \Theta_i \mapsto \Delta(A_i)$ is a mapping from agent $i$'s types to distributions over agent $i$'s actions.
The set of all strategies for agent $i$ is denoted by $S_i$.
By $s_i(a_i \given \theta_i)$, we indicate the probability that agent $i$ takes action $a_i$ under $s_i$ given that agent $i$'s type is $\theta_i$.
Similar to action profiles, a strategy profile $s = (s_1, \dots, s_n) \in S$ is a vector of strategies, where $S = S_1\times \dots \times S_n$ is a set of all possible strategy profiles.

\noindent
\textbf{Expected utilities.}
In Bayesian games, there are two main sources of uncertainty: (i) types and (ii) actions.
Types are drawn from the prior probability distribution, $P$, and actions are taken based on agents' strategies.
To capture both sources of uncertainty, the ex ante expected utility of agent $i$ is modeled as follows.
\begin{align*}
\EU_i(s) = & \EX_{a,\theta}[u_i(a, \theta)] = p_b \; \EX_{a, \theta}[u_i(a,\theta) \given \theta_{i,0}=1] \\
& + (1-p_b) \; \EX_{a, \theta}[u_i(a,\theta) \given \theta_{i,0}=0] .
\end{align*}
The expectation is taken with respect to $\theta$ and $a  \sim s(\cdot \given \theta)$.
This formula models the expected utility of agent $i$ before the start of the game and before the agent observes its type.

Given the defined expected utility model, we can define the set of agent $i$'s best responses to strategy profile $s_{-i}$ as:
\[
\BR_i(s_{-i}) = \underset{s_i \in S_i}{\argmax} ~~ \EU_i((s_i, s_{-i})).
\]
Intuitively, a best response is a strategy which provides the highest expected utility given the strategy of others.
Note that there may be more than one strategy that maximizes agent $i$'s expected utility for a given $s_{-i}$.

\noindent
\textbf{Bayesian Nash equilibrium.}
As discussed before, a strategy is a full contingency plan.
Agents \textbf{simultaneously} choose their strategies \textbf{before} the start of the game and do not change their adopted strategies during the game. 
Once the game starts, each agent observes its type and acts as prescribed by its strategy.
Agents strategies form a Bayesian Nash equilibrium (BNE) when the strategy of each agent is a best response to the strategies adopted by other agents.
Formally, a strategy profile $s^{*}$ is a BNE if and only if $s^{*}_i \in \BR_i(s^{*}_{-i})$, for all $i$. 
Informally, in a BNE, agent $i$ does not have any incentive to unilaterally change its strategy from $s^*_i$ if it knows that other agents have fixed their strategies to $s^{*}_{-i}$.

\section{Incentive Analysis in Algorand}
\label{sec:analysis}
In this section, we formulate agents' utilities and study their BNE strategies.
We first consider Algorand's original reward scheme.
We show that under this reward scheme, cooperation is not a BNE strategy.
We then propose a novel reward scheme and show that under certain conditions, our proposed reward scheme incentivizes all non-Byzantine agents to cooperate regardless of their type.

\subsection{Algorand's Original Reward Scheme}
Algorand's original reward scheme is called Participation Rewards \cite{governance-rewards}.
Under this reward scheme, \emph{The Algorand Foundation} distributes a fixed amount of cryptocurrency assets as a reward among all agents.
Agents are assigned sub-nodes in proportion to the balance of their accounts.
Under Algorand's original reward scheme, the fixed reward, $R$, is distributed equally among all sub-nodes.
If agent $i$ is assigned $w_i$ sub-nodes, then agent $i$'s reward, $R_i$, is equal to $R \; w_i / W$, where $W$ is the total number of sub-nodes in the system.
The first advantage of this reward scheme is its simplicity: it is easy to implement, and it is easy to explain to agents how they are rewarded.
The most important advantage of Algorand's original reward scheme is that it provides proportional rewards.

\begin{definition}{Proportional rewards.}
Let $R_i$ denote the expected reward of agent $i$.
A reward scheme provides proportional rewards if for any agent $i$ and $j \in N$, $R_i / R_j = w_i / w_j$.
\end{definition}

The proportional rewards property ensures that the expected fraction of assets controlled by the adversary does not increase by the action of the reward scheme.
Although Algorand's original reward scheme is simple and provides proportional rewards, it suffers from a key drawback: it fails to prevent free riding.
Agents do not have any incentive to cooperate as they receive their rewards irrespective of their cooperation.

\begin{theorem}
Let $s^*$ be a strategy profile where for each agent $i$, if $\theta_{i,0} = 0$, then $s_i^*(C \given \theta_i) = 1$, and otherwise, $s_i^*(M \given \theta_i) = 1$.
Under Algorand's original reward scheme, $s^*$ is not a BNE.
\end{theorem}

\begin{proof}
Suppose agent $i$ is a non-Byzantine agent (i.e., $\theta_{i,0} = 0$).
If agent $i$ defects, it receives its rewards without incurring any costs.
Formally, $u_i(a,\theta) = R / n$ for all $\theta \in \Theta$ if $a_i = D$%
\footnote{We present the proof for the case where all agents are assigned a single sub-node.
Our proof easily extends to the case where agents are assigned different number of sub-nodes.}.
Let $s^\prime_i$ be a strategy that chooses $D$ regardless of the type (i.e., $s^\prime_i(D \given \theta_i) = 1$ for all $\theta_i \in \Theta_i$).
It can be easily shown that $\EU_i((s^\prime_i, s^*_{-i})) = (1 - p_b) \; R / n$.
If a non-Byzantine agent $i$ cooperates, it receives its reward but incurs some strictly positive costs.
This means that $u_i(a,\theta) < R / n$ if $a_i = C$.
Let $a^*(\theta)$ be an action profile where $a^*_i$ is $C$ if $\theta_{i,0} = 0$ and $M$ otherwise.
The expected utility of agent $i$ for $s^*$ is:
\begin{align*}
\EU_i(s^*) & = (1-p_b) \; \EX_{\theta}[u_i(a^*(\theta),\theta) \given \theta_{i,0}=0] \\
& < (1 - p_b) \; R / n = \EU_i((s^\prime_i, s^*_{-i})).
\end{align*}
This implies that $s^*_i$ is not a best response to $s^*_{-i}$.
Therefore, $s^*$ is not a BNE under Algorand's original reward scheme.
\end{proof}

\subsection{Incentive-compatible Reward Scheme (\name)}
\label{sec:proposed-reward}

To address the free-rider problem, we propose \name, a novel incentive-compatible reward scheme for Algorand.
Under \name, the Algorand Foundation distributes rewards among agents based on their cooperation.
The cooperation of committee members can be easily tracked as their messages are guaranteed to reach all other agents.
However, tracking the cooperation of the agents that are not selected as a committee member is challenging as they do not initiate any messages.
To address this challenge, our proposed reward scheme requires committee members at step $k = 2,\dots,\kmax$ to include with their vote the identities of the agents from which they have received a valid message at step $k - 1$.

Let $R^c(k)$ and $R^{b}(k)$ denote a fixed baseline reward and a fixed committee reward at step $k$, respectively.
Under \name, a cooperating agent $i$ receives a committee reward of $R^c(k)$ if it is selected as a committee member at step $k$.
Additionally, if agent $i$'s identity is included in the vote generated by agent $j \in N_i$ at step $k$, then agent $i$ receives a baseline reward of $R^{b}(k) / n_{rs}$.
We assume that Byzantine agents do not include the identity of non-Byzantine agents in their vote when they are selected as committee members.
In other words, non-Byzantine agents do not receive any baseline reward for propagating messages to their Byzantine peers.
We make this assumption to calculate a lower bound on the expected utility of non-Byzantine agents.
If Byzantine agents include the identity of non-Byzantine agents, the expected utility of non-Byzantine agents can only increase.
Given this assumption, the total reward of a cooperating agent $i$ at step $k$ given an action profile $a$ and a type vector $\theta$ can be formulated as:
\begin{align}
\label{equ:our_reward}
R_i(a, \theta, k) & = (R^{b}(k) / n_{rs}) \sum_{j \in N_i} \II(\theta_{j,k} \le p_c, a_j = C) \\
& + R^c(k) \; \II(\theta_{i,k} \le p_c).
\end{align}

To show that our reward scheme prevents the free-rider problem, we first prove that any non-Byzantine agent cannot unilaterally change the outcome of each step by defecting.
Given that, we formulate the expected utility of each non-Byzantine agent assuming that all other non-Byzantine agents cooperate regardless of their type.
We then prove that if certain conditions are met, cooperation is a best response for a non-Byzantine agent when other non-Byzantine agents cooperate.
This then shows that cooperation is a BNE strategy for non-Byzantine agents under \name.

\begin{lemma}
\label{lemm:chernoff}
Let $\NR_k$ and $\NB_k$ be random variables indicating the number of non-Byzantine and the number of Byzantine agents selected as committee members at any step $k$, respectively.
Let $\mu_r = \EX(\NR_k)$, $\mu_b = \EX(\NB_k)$, $\mu = \mu_b + \mu_r / 2$, $\delta_r = 1 - T / (1 - p_b)$, and $\delta = 2 \times T / (1 + p_b) - 1$.
Then we have:
\begin{itemize}
 \item $\EP(\NR_k \le T \; \tau) \le e^{-\mu_r \; \delta_r^2 / 2}$, and
 \item $\EP(\NB_k + \NR_k / 2 \ge T \; \tau) \le e^{-\mu \; \delta^2 / (2 + \delta)}$.
\end{itemize}
\end{lemma}

\begin{proof}
Let $X_{i,k}$ be a random variable that takes value 1 if agent $i$ is selected as a committee member and agent $i$ is not corrupted by the adversary, and takes value 0 otherwise.
Similarly, let $Y_{i,k}$ be a random variable that takes value 1 if agent $i$ is selected as a committee member and agent $i$ is corrupted by the adversary, and takes 0 otherwise.
We can write $\NR_k = \sum_i X_{i,k}$, and $\NB_k = \sum_i Y_{i,k}$.
We have $\EX(X_{i,k}) = (1 - p_b) \; p_c$, and $\EX(Y_{i,k}) = p_b \; p_c$.
Therefore, $\mu_r = \EX(\NR_k) = \sum_i \EX(X_{i,k}) = (1 - p_b) \; \tau$, and $\mu_b = \EX(\NB_k) = \sum_i \EX(Y_{i,k}) = p_b \; \tau$.
Since $\delta_r \ge 0$, according to Chernoff bound, the following inequality holds for any $k$.
\begin{align*}
\EP(\NR_k \le T \; \tau) & = \EP(\NR_k \le (1-\delta_r) \; (1 - p_b) \; \tau) \\
& = \EP(\NR_k \le (1-\delta_r) \; \mu_r) \le e^{-\mu_r \; \delta_r^2 / 2}.
\end{align*}
Similarly, since $\delta \ge 0$, we have the following inequality for any $k$.
\begin{align*}
\EP(\NB_k + \NR_k / 2 \ge T \; \tau) & = \EP(\NB_k + \NR_k / 2 \ge (1 + \delta) \; (1 + p_b) \; \tau / 2) \\
& = \EP(\NB_k + \NR_k / 2 \ge (1 + \delta) \; \mu) \\
& \le e^{-\mu \; \delta^2 / (2 + \delta)}.
\end{align*}
\end{proof}

We assume that $p_b$, $\tau$, and $T$ are set such that $\NR_k > T \; \tau$ and $\NB_k + \NR_k / 2 < T \; \tau$ are true with overwhelming probability.
For example, using \reflemm{lemm:chernoff}, with $p_b = 0.2$, $\tau = 4000$ and $T = 0.7$, the probability that  $\NR_k \le T \; \tau$ is less than $10^{-10}$ and the probability that $\NB_k + \NR_k / 2 \ge T \; \tau$ is less than $10^{-13}$.
The two inequalities imply that more than 2/3 of the selected committee members at each step are non-Byzantine agents with overwhelming probability.

\begin{proposition}
\label{thm:expected_steps}
Consider any agent $i \in N$.
Suppose that the strategy profile of all agents except agent $i$ is $s_{-i}^*$ where for each agent $j \in N \setminus i$, if $\theta_{j,0} = 0$, then $s_j^*(C \given \theta_j) = 1$, and $s_j^*(M \given \theta_j) = 1$ otherwise.
In a large system (i.e., $n \rightarrow \infty$), the safety and liveness guarantees of the Algorand protocol are met with high probability regardless of the strategy of agent $i$.
\end{proposition}

\begin{proof}
Assume that $p_b$, $\tau$, and $T$ are set such that $\NR_k > T \; \tau$ and $\NB_k + \NR_k / 2 < T \; \tau$ with overwhelming probability.
Consider a new system consisting of all agents except agent $i$.
For large systems, it is easy to modify \reflemm{lemm:chernoff} to show that the two inequalities still hold for the new system with the same $p_b$, $\tau$, and $T$.
Moreover, the Chernoff bound can be applied to show that with $p_b \le 1/3$, the network is synchronous with high probability regardless of agent $i$'s cooperation.
Given that the network is synchronous, and the two inequalities hold with high probability, Theorem 1 from \cite{CM2019} can be applied to the new system to guarantee safety and liveness.
\end{proof}

We next formulate the expected utility of each agent.
The Algorand protocol implements a randomized algorithm.
As discussed in \refsec{sec:algorand-prot}, the algorithm runs in multiple steps.
We use $K(a, \theta)$ to denote the total number of steps it takes the algorithm to complete as a function of agents' types and their actions.
Given $a$ and $\theta$, the total utility of agent $i$ with $\theta_{i,0} = 0$ can be formulated as:
\[
u_i(a, \theta) = \sum_{k=1}^{K(a, \theta)} u_{i}(a, \theta, k), 
\]
where $u_{i}(a, \theta, k)$ is $R_i(a, \theta, k) - C_i(a, \theta, k)$%
\footnote{In \refequ{equ:cost}, $C_i$ depends on $\theta$ through $.x_{i,k}$, and it is formulated assuming that agent $i$ cooperates.}
if $a_i = C$ and 0 otherwise.

\begin{lemma}
\label{lemm:expected-utility}
Consider any agent $i \in N$.
Suppose that the strategy profile of all agents except agent $i$ is $s_{-i}^*$ where for each agent $j \in N \setminus i$, if $\theta_{j,0} = 0$, then $s_j^*(C \given \theta_j) = 1$, and $s_j^*(M \given \theta_j) = 1$ otherwise.
Suppose further that agent $i$ adopts strategy $s_i$ which plays $M$ if $\theta_{i,0} = 1$ and plays $C$ with probability $s_c$ and $D$ with probability $1 - s_c$ if $\theta_{i,0} = 0$.
Let $a^*(\theta)$ be an action profile where $a^*_i$ is $C$ if $\theta_{i,0} = 0$ and $M$ otherwise.
Define $K(\theta) = K(a^*(\theta), \theta)$.
The expected utility of agent $i$ is:
\[
\EU_i((s_i, s^*_{-i})) = s_c \; (1-p_b) \; \sum_{\ell=1}^{\kmax} \EP(K(\theta) = \ell) \sum_{k=1}^{\ell} u(k),
\]
where $u(k) = R^{b}(k) \; p_c \; (1-p_b) - C^b(k) + (R^c(k) - C^c(k)) \; p_c$.
\end{lemma}

The proof is deferred to \refapp{app:utility-proof}.
Next, we prove that $s^*$ is a BNE under certain conditions.
These conditions are derived simply by requiring that the expected rewards should outweigh the expected costs. 
\begin{theorem}
\label{thm:BNEByzantine}
Let $s^*$ be a strategy profile where for each agent $i$, if $\theta_{i,0} = 0$, then $s_i^*(C \given \theta_i) = 1$, and otherwise, $s_i^*(M \given \theta_i) = 1$.
Given $p_c$, $p_b$, $C^b(k)$, and $C^c(k)$ for  $k = 1, \dots, \kmax$, $s^*$ is a BNE under \name if:
\[
R^{b}(k) (1-p_b) + R^c(k) \ge C^b(k) / p_c + C^c(k),  ~~ \forall k = 1, \dots, \kmax.
\]
\end{theorem}

\begin{proof}
To prove that $s^*$ is a BNE, it suffices to show that $s_i^*$ is a best response to $s_{-i}^*$ for all agent $i \in N$.
Recall that $s_i^*$ is a best response to $s_{-i}^*$ when $\EU_i((s_i^*, s^*_{-i})) \ge \EU_i((s_i, s^*_{-i}))$ for all $s_i \in S_i$.
If $R^{b}(k) \; (1-p_b) + R^c(k) \ge C^b(k) / p_c + C^c(k)$ for all $k = 1, \dots, \kmax$, then $u(k)$ in \reflemm{lemm:expected-utility} is greater than or equal to zero for all $k = 1, \dots, \kmax$.
Consequently, $\EU_i((s_i, s^*_{-i}))$ is maximized when $s_c = 1$.
In this case, we have $s = s^*$ which means $s^*$ is a BNE.
\end{proof}
\section{Implementation Details}
\label{sec:impl}

\subsection{Gossip Protocol in \name}
To implement \name, we make three main modifications to Algorand's default gossip protocol.
First, we require the randomness of the peer-selection mechanism to be verifiable (e.g., through verifiable pseudo-random peer selection \cite{LCW+2006}).
This requirement prevents the adversary from gaining unauthorized awards.
Byzantine nodes cannot be rewarded for propagating messages to nodes that are not among their randomly selected peers.
Byzantine committee members also cannot refer other Byzantine nodes that are not randomly connected to them.

Second, we require nodes to disable selective propagation.
Selective propagation is an optimization technique that  prevents nodes from propagating low-priority block proposals \cite{GHM+2017, CM2019}.
Although this technique reduces network congestion, it prevents low-priority block proposals from reaching all nodes in the gossip network.
This in turn prevents the Algorand Foundation from tracking cooperation of some committee members at step 1.
One optimization that could be implemented to replace selective propagation is to only send the committee member's credential for the low-priority block proposals without sending the entire block proposal.

Third, we require nodes to track the identity of all nodes from which they have received a valid message, even if the message is a duplicate message.
For example, suppose that agent $i$ receives message $m$ at step $k$ first from agent $j$ and later from agent $j^\prime$.
Agent $i$ propagates message $m$ to its peers only the first time it receives it from agent $j$.
However, it saves the identity of both agents $j$ and $j^\prime$ as propagators for $m$ at step $k$.
If agent $i$ becomes a committee member at step $k + 1$, it includes the identity of both agents $j$ and $j^\prime$ in its generated vote.
This allows the Algorand Foundation to not only track the cooperation of the committee members but also the cooperation of their peers.
Since committee members are selected randomly, their cooperating peers can be considered random samples of non-voting agents that cooperate.

\subsection{Consideration of Assets in \name}
\label{sec:imp-assets}
Algorand assigns sub-nodes to each node in proportion to the balance of its account.
For simplicity, in our analysis so far, we have considered a single sub-node per node.
We now discuss the necessary changes needed to allow multiple sub-nodes per node.

\noindent
\textbf{Cryptographic sortition.}
The sortition algorithm can be easily extended to consider node with more than one sub-node.
Suppose that node $i$ has $w_i$ sub-nodes.
A simple way to extend sortition is for node $i$ to run the sortition algorithm (see \refalg{alg:algorand}) on each of its $w_i$ sub-nodes.
Although simple, this method is computationally expensive.
An alternative method, proposed in \cite{CM2019}, is to use the inverse transform sampling.
In this method, the interval of $[0, 1)$ is partitioned into consecutive intervals of the following form.
$$I^{M}_{w_i, p} = \big[B(M; w_i, p), B(M+1; w_i, p)\big), ~~~~~~~~~ \forall M \in \{0, 1,\dots,w_i - 1\}.$$
If $.x_i$ falls in the $I^{M}_{w_i,p}$ interval, node $i$ has $M$ selected sub-nodes.

\noindent
\textbf{Adversary model.}
So far, we have assumed that the adversary corrupts each node with probability $p_b < 1/3$.
We can extend this assumption by allowing the adversary to corrupt each sub-node with probability $p_b$.
Suppose that the total number of sub-nodes is $W = \sum_{i\in N} w_i$.
We define $W_{-i} = \sum_{j \in N \setminus i} w_j$ for all agents $i$.
For our results to hold under the new adversary model, we require $W_{-i}$ to go to infinity as $n$ goes to infinity.
We further require that $w_{i} / W < (1 - 3 \times p_b)$ for all nodes $i$.
This requirement ensures that if any single node is removed from the set of nodes, in expectation, the adversary does not control more than one-third of the remaining sub-nodes.
For example, if $p_b = 0.3$, then no single agent should have greater than or equal to one-tenth of all sub-nodes.
Under these assumptions, it can be easily verified that our results in \refsec{sec:analysis} hold for the new adversary model because non-Byzantine sub-nodes still constitute more than $2/3$ of the committee with high probability.

\noindent
\textbf{Gossip network and protocol.}
When each agent has more than one sub-node, instead of forming the random gossip network among nodes, we require the gossip network to be constructed among sub-nodes.
This means that nodes need to select a subset of random sub-nodes as peers for each of their sub-nodes.
This can be achieved by including indexes of sub-nodes as an extra input for the verifiable pseudo-random peer selection procedure \cite{LCW+2006}.
When a sub-node propagates a message to another sub-node, it is required to include the indexes of both the sender and the receiver sub-nodes in the message.
When a node receives a message, it will verify if the sender sub-node and the receiver sub-node are indeed connected.

\noindent
\textbf{\name rewards.}
The simplest way to extend \name to consider multiple sub-nodes per node is to distribute rewards on a per-sub-node basis. 
In other words, IRS rewards each sub-node independently.
In this way, the expected total reward of a node is the sum of the expected rewards of its sub-nodes.
The expected total cost depends on the implementation of the sortition process and the gossip protocol. 
With the simple extensions presented above, the total expected cost is the sum of the expected costs of its sub-nodes.

\section{Experiments}
\label{sec:exp}

In this section, we estimate the costs of running an Algorand node for two real-world deployment scenarios.
Based on the costs, we identify the baseline reward, $R^b$, and committee reward, $R^c$, for the two scenarios to guarantee participation as a BNE strategy.
We also simulate the extended gossip protocol outlined in \refsec{sec:impl} for \name and measure its overheads in terms of computation, storage and networking.

\subsection{Estimated Costs and Rewards}
Compared to PoW blockchains, Algorand nodes require much less computation power to participate in the consensus protocol.
Any tentative participant can run an Algorand node with commodity machines.
We consider two deployment scenarios.
First, the participant runs an Algorand node on a virtual machine from a cloud service provider such as Amazon Web Services (AWS) \cite{AWS}.
Second, the participant runs the Algorand protocol on a personal computer.
We estimate the costs in Algo%
\footnote{Algo is the unit of Algorand's cryptocurrency.} 
($1$ Algo $= 0.34$ USD).
\begin{table}[t!]
\centering
    \begin{tabular}{|c | c | c|} 
        \hline
        Comp. \& Storage Cost & $72.54$ \$/mon & $8.2 \times 10^{-5}$ Algo/s \\ 
        \hline
        Networking Cost & $0.09$ \$/GB & $0.2647$ Algo/GB \\
        \hline
    \end{tabular}
\vspace{2ex}
\caption{Costs of an AWS EC2 instance}
\label{table:AWSCosts}
\vspace{-4mm}
\end{table}

\noindent
\textbf{Cloud-hosted Algorand node.}
Consider an AWS EC2 instance with 4 vCPU, 8GB memory, and 256GB SSD storage.
The costs associated with renting such instance are listed in \reftab{table:AWSCosts}.
The average size of each block in Algorand ledger is about $40$KB \cite{Algoscan}.
We estimate the size of a vote message to be about $1$KB based on the official implementation of Algorand in Go programming language \cite{GoAlgorand}.
On average, it takes about $5$ seconds for the Algorand protocol to commit a block \cite{AlgoMetric}.
We assume that each round takes $5$ steps, which means that each step takes about $1$ second.
We further assume that each Algorand node has $8$ peers in the gossip network.
This is the default value in Algorand \cite{GHM+2017}.
According to current specification of Algorand \cite{MainnetSpec}, $p_b = 1/5$ and $\tau = 20, 2990, 1500, 5000, 5000, \dots, 5000$ for step 1 to $\kmax$.
As of July 24th, 2022, according to the AlgoExplorer \cite{AlgoExplorer}, there are $550$ nodes with a total of more than $25$ billion sub-nodes actively participating in running the Algorand protocol.

We estimate the baseline cost at each step as computation and storage costs plus communication costs.
Computation and storage costs are listed in \reftab{table:AWSCosts} and are fixed across steps.
The communication cost per step is calculated as a linear function of the number of neighbours per node and the per GB networking cost listed in \reftab{table:AWSCosts}.
Similarly, the committee cost at each step is estimated as the sum of computation and storage costs and communication costs.
The computation and storage costs are covered in our estimated baseline cost (the same rented AWS EC2 instance executes both baseline and committee related potions of the Algorand protocol).
The communication cost per step is calculated similarly as before.
Given the costs, we derive the baseline rewards and committee rewards at each step by ensuring that $R^{b}(k) (1-p_b) \ge C^b(k) / p_c$ and $R^c(k) \ge C^c(k)$ according to \refthm{thm:BNEByzantine}.
The estimated costs and derived \name rewards per block are listed in \reftab{table:AWSEstimation}.

\begin{table}[t!]
\centering
    \begin{tabular}{|c | c | c | c | c|} 
        \hline
        Step & $C^b$ & $R^b$ & $C^c$ & $R^c$ \\ 
        \hline
        1 & $1.77 \times 10^{-3}$ & 18603 & $8.47 \times 10^{-5}$ & $8.5 \times 10^{-5}$ \\
        \hline
        2 & $6.4 \times 10^{-3}$ & 133891 & $2.12 \times 10^{-6}$ & $2.2 \times 10^{-6}$ \\
        \hline
        3 & $3.26 \times 10^{-3}$ & 20408 & $2.12 \times 10^{-6}$ & $2.2 \times 10^{-6}$ \\
        \hline
        4 to $\kmax$ & 0.011 & 66823 & $2.12 \times 10^{-6}$ & $2.2 \times 10^{-6}$ \\
        \hline
    \end{tabular}
\vspace{2ex}
\caption{Estimated costs and rewards (Algo/block) on AWS.}
\label{table:AWSEstimation}
\vspace{-4mm}
\end{table}

\noindent
\textbf{Self-hosted Algorand node.}
Next, we consider the costs and associated \name rewards when an Algorand node runs on a personal computer (PC).
We consider the energy consumption of a typical PC to be about $200$W \cite{computerEnergy}.
We consider an average energy price of $0.0944$ USD/kWh \cite{OhioEnergyPrice}.
We further consider an average price of an Internet plan with unlimited traffic to be around $100$ USD/month.
Given these costs, the overall cost of running an Algorand node on a personal computer is about $0.00013$ Algo/s.
Similar to our cost estimation for the AWS scenario, the baseline cost covers the total computation and storage costs.
Since there is no limit on the traffic, we can simply set $C^b = 0.00013$ Algo and $C^c = 0$ Algo.
The baseline \name rewards are listed in \reftab{table:PCEstimation}.

\begin{table}[t!]
\centering
    \begin{tabular}{|c | c | c | c | c|} 
        \hline
         & Step 1 & Step 2 & Step 3 & Steps 4 to $\kmax$ \\ 
        \hline
        $R^b$ & 1362 & 2627 & 815 & 815\\
        \hline
    \end{tabular}
\vspace{2ex}
\caption{Estimated baseline rewards (Algo/block) on PC.}
\label{table:PCEstimation}
\vspace{-4mm}
\end{table}

\noindent
\textbf{Comparison against Algorand's original reward scheme.}
Algorand's Participation Rewards only distributes about $20$ to $30$ Algos for each block.
In comparison, our calculated rewards might seem prohibitively large for a practical implementation and deployment of \name.
We note that most blocks are generated within the first 5 steps. 
This means that for most blocks, no reward is given for steps 6 to $\kmax$. 
Second, our calculations are based on the assumption that each agent has a single sub-node, which leads to a worst-case analysis (unlike expected rewards, expected costs do not necessarily increase linearly as a function of sub-nodes).
In the real-world deployment of Algorand, each agent has on average 40 million sub-nodes. 
Restricting the analysis to such agents could lead to more feasible rewards.
Finally, in real-world settings, some agents are altruistic (i.e., they participate regardless of the rewards). 
Inclusion of such agents in the analysis would lead to lower rewards.
Extending IRS to include altruistic agents would be an interesting future work, both theoretically and in practice. 

\subsection{Overhead of Gossip Protocol in \name}
The overhead of the extended gossip protocol in \name consists of three main parts: 
(i) computational overhead of running verifiable random peer selection, 
(ii) networking overhead of disabling selective propagation, and 
(iii) storage overhead of storing identities of all valid message propagators.

\noindent
\textbf{Computation overhead.}
The verifiable random peer selection procedure randomly connects sub-nodes to each other.
To implement this procedure, the inverse transform sampling technique in Algorand's sortition procedure can be used (see \refsec{sec:imp-assets}).
To measure the overhead of this implementation, we run the benchmark of the sortition function in the official implementation of Algorand \cite{GoAlgorand} on a machine powered by AMD EPYC 7H12 processors.
The average execution time to run the sortition once is about $0.2$ ms.
Algorand reconstructs the gossip network for each new block \cite{GHM+2017}, so the overall computation overhead is then $0.0002 \times (N - 1) \approx 0.11$ seconds per block for $N = 550$.
Given that each block is committed to the ledger every 5 seconds, this is an overhead of 2.2\%.
One way to reduce this overhead is to reduce the network reconstruction frequency (e.g., reconstruct the gossip network every 4 or 5 blocks).

\noindent
\textbf{Network and storage overhead.}
To measure the network and the storage overheads, we implement a simulator of the gossip protocol in \name.
The simulator has 132 lines of Python code and is available at \url{https://anonymous.4open.science/r/IRS-Simulator/}.
We run our simulator on real-world data extracted from AlgoExplorer \cite{AlgoExplorer}.
We randomly connect each sub-node to 8 other sub-nodes.
According to our simulation results, on average, each Algorand node is connected to $252$ peers (each node has on average more than 45 million sub-nodes).
Moreover, our results show that each node has to gossip the hashes of low-priority blocks $18$ times (on average 20 selected sub-nodes as block proposers).
Assuming that the size of each block hash is $1$KB, the network bandwidth overhead for each node is about $252 \times 18 \times 1$KB $= 4536$KB per committed block.
For storage overhead, each node needs to store the public keys of all valid message propagators.
The size of the public key in Algorand is $32$ bytes \cite{GoAlgorand}, so the storage overhead is about $252 \times 32$B $= 8.064$KB per committed block.
Regarding networking and storage, \name has limited carbon footprint because networking and storage are energy disproportional (i.e., the energy consumption does not proportionally change as a function of network/storage utilization). 
\section{Related Works}

\noindent
\textbf{Incentives in blockchains.}
Incentive compatibility of reward schemes is crucial to the security of blockchains.
Eyal and Sirer \cite{ES2014} prove negative results with regard to incentive compatibility of Bitcoin.
They model the competitive mining among the miners as a strategic game and propose a novel mining strategy called selfish mining.
The authors examine the profitable threshold of selfish mining, which is the computation power needed to gain more revenue.
The authors then propose a modification to the Bitcoin chain-selection protocol to increase the profitable threshold of selfish mining to 1/4.
Selfish mining has been studied extensively ever since \cite{SSZ2016,NKM+2016,ZZK2020,HZJ+2021,PS2017}.

\noindent
\textbf{Incentives for information propagation.}
Most permissionless blockchains \cite{Nakamoto2008,wood2014ethereum,GHM+2017} rely on a peer-to-peer gossip protocol to propagate transactions.
However, these protocols do not provide an incentive for the nodes in the blockchains to participate in propagating the transactions.
In fact, Babaioff et al. \cite{BDO+2012} argue that these protocols incentivize nodes not to propagate transactions.

The authors in \cite{BDO+2012} propose a reward scheme to incentivize transaction propagation.
They prove that under the proposed reward scheme, the strategy where all nodes propagate transactions and do not create fake identities (Sybil attacks) is a Nash equilibrium.
In addition, the proposed scheme guarantees that most of the nodes in the network will be aware of the transaction.
The additional rewards required to implement the scheme are a constant in expectation and the user only needs to send the transaction to a small number of nodes in the beginning.
A main drawback of the model in \cite{BDO+2012} is that it is highly restricted.
The model only considers networks in the form of a forest of $d$-ary directed trees with height $H$ and assumes that each node has the same hashing power.

Ersoy et al. \cite{ERE+2018} propose another incentive mechanism for transaction propagation under a network model with minimal restrictions.
The resulting reward scheme encourages propagating without creating fake identities.

\noindent
\textbf{Incentives in committee-based protocols.}
Amoussou-Guenou et al. \cite{ABP+2020,ABP+2020a} analyze a simplified committee-based consensus protocol.
They propose to make some committee members pivotal such that they have incentives to participate.
In this case, a rational node can unilaterally determine the result of the consensus protocol.
If a pivotal node does not follow the protocol, consensus is not be reached and a penalty is applied.
As a result, all pivotal nodes participating in the protocol as required by the protocol becomes a Nash equilibrium.
Nevertheless, the solution is not practical in Algorand because the assumptions in \cite{ABP+2020,ABP+2020a} do not hold.
For example, the solution proposed by \cite{ABP+2020a} assumes that all nodes are treated equally and have the same voting power, which is not the case in Algorand.
Also, the solution assumes that all the committee members are ordered by publicly known indexes, while the membership is private in Algorand until nodes publish it.
In addition, each node communicates with all other nodes directly in the simplified protocol while Algorand adopts a gossip protocol to disseminate messages.

Fooladgar et al. \cite{FMJ+2020} analyze the Algorand protocol as a static non-cooperative game.
They demonstrate that all nodes participating the protocol is not a Nash equilibrium under Algorand's original reward scheme.
If all other nodes participate in the protocol, a node can free-ride to get its reward without paying the costs of running an Algorand node.
To address the free-rider problem, the authors propose a new reward scheme where only participating nodes are rewarded.
In this work, the analysis is limited to the rational-agent-only case.
In addition, the proposed reward scheme only incentivizes committee members and
fails to incentivize message propagation by all nodes.
Under the proposed method, only nodes whose participation is necessary (for the network to remain synchronous) are incentivized to propagate messages. 

We extend the prior work in two main aspects.
First, our model preserves important features of the Algorand protocol, such as randomized membership selection and private membership information.
Second, we propose a new reward scheme and prove that all nodes participating faithfully is a Bayesian Nash equilibrium under certain conditions even in the presence of an adversary.

\section{Conclusion \& Future Works}

In this paper, we model the Algorand protocol as a Bayesian game.
We propose a reward-distribution scheme and derive necessary conditions such that all nodes participating in the protocol is a Bayesian Nash equilibrium.
Our work highlights some open problems in designing incentive-compatible reward schemes for randomized, committee-based consensus protocols that depend on nodes for message propagation.
In particular, like Algorand's original reward scheme, our proposed reward scheme relies on a central authority to collect messages to identify the role of nodes and to verify the participation of nodes.
A more decentralized reward scheme is preferred where nodes in Algorand network determine how the rewards are distributed without a central authority.
Moreover, although our model captures some important features of Algorand, modeling the protocol as a sequential game is an interesting direction for future work.
Finally, our analysis is based on the assumption that the network is synchronous.
Extending the analysis to a setting where the network is not synchronous would be a possible future direction.

\begin{acks}
This work was partially supported by the NSERC-RGPIN-2019-04936, CFI-JELF-38850, ORF-RI-38850, Ripple, and NSERC Discovery grants.
\end{acks}

\bibliographystyle{ACM-Reference-Format} 
\bibliography{bib}

\newpage
\appendix
\section{Notations}
\label{app:notations}
In this appendix, we provides a summary of all the parameters we used in this paper.
\reftab{tab:notations} details the parameters used to describe the Algorand protocol, costs and rewards, and the Algorand game.
\begin{table}[!htbp]
\centering
\caption{Parameters for costs and rewards}
\label{tab:notations}
\begin{tabular}{ ll }
    \toprule
    \textbf{Parameter} & \textbf{Description} \\
    \toprule
    $\sigma_k$ & Committee credential of a node \\
    $x_k$ & Hash of $\sigma_k$ \\
    $.x_k$ & Interpretation of $x_k$ as a number between 0 and 1 \\
    $\tau$ & Expected number of sub-nodes selected at each step \\
    $T$ & Fraction of $\tau$ that defines Algorand's voting threshold \\
    $p_c$ & Probability that an agent is selected as a committee member \\
    $p_b$ & Probability that an agent is corrupted by the adversary \\
    $K_{\max}$ & Maximum number of steps to add a new block \\
    $C_i^b(k)$ & Baseline cost incurred by agent $i$ at step $k$ \\
    $C^c(k)$ & Committee cost incurred by any committee member at step $k$ \\
    $C_i(k)$ & Total cost incurred by agent $i$ at step $k$ \\
    $R_i^{b}(a, \theta, k)$ & Baseline reward of cooperating agent $i$ at step $k$ \\
    $R_i^c(a, \theta, k)$ & Committee reward of cooperating agent $i$ at step $k$ \\
    $R_i(a, \theta, k)$ & Total reward of cooperating agent $i$ at step $k$ \\
    $N$ & Set of all agents \\
    $w_i$ & Number of sub-nodes of agent $i$ \\
    $W$ & Total number of sub-nodes for all agents ($W = \sum_{i \in N} w_i$) \\
    $N_i$ & Set of agent $i$'s randomly selected peers \\
    $\hat{N}_i$ & Set of agents of which agent $i$ is a randomly selected peer \\
    $n_{rs}$ & Number of each node's randomly selected peers ($n_{rs} = \vert N_i \vert$) \\
    $a_i$ & Action taken by agent $i$ \\
    $A_i$ & Set of actions available to agent $i$ \\
    $a$ & Joint action profile of all agents ($a = (a_1,\dots, a_n)$) \\
    $A$ & Set of possible joint action profiles \\
    $\theta_{i,0}$ & Binary variable indicating whether agent $i$ is Byzantine or not \\
    $\theta_{i,k}$ & Number of agent $i$'s sub-nodes selected at step $k > 0$ \\
    $\theta_i$ & Type vector of agent $i$ ($\theta_i = (\theta_{i,1},\dots,\theta_{i,K_{\max}})$) \\
    $\Theta_i$ & Type space of agent $i$ \\
    $\theta$ & Joint type vector of all agents ($\theta = (\theta_{1},\dots,\theta_{n})$) \\
    $\Theta$ & Set of possible joint type vectors \\
    $s_i$ & Strategy of agent $i$ \\
    $s_i(a_i \given \theta_i)$ & Probability assigned to action $a_i$ by $s_i$ given $\theta_i$ \\
    $S_i$ & Set of all strategies for agent $i$ \\
    $s$ & Joint strategy profile for all agents \\
    $S$ & Set of possible joint strategy profiles \\
    $EU_i(s)$ & Expected utility of agent $i$ given strategy profile $s$ \\
    $BR_i(s_{-i})$ & Best response of agent $i$ given $s_{-i}$ \\    
    \bottomrule
\end{tabular}
\end{table}
\newpage
\section{Proof of Lemma (6)}
\label{app:utility-proof}

If agent $i$ is corrupted by the adversary, it strictly prefers action $M$, which leads to a payoff of zero as defined in \refsec{sec:algorand-game}.
If agent $i$ is not corrupted, it takes action $D$ with probability $1 - s_c$ for a payoff of zero.
We can write the expected utility of agent $i$ for $(s_i, s^*_{-i})$ as:
\[
\EU_i((s_i, s^*_{-i})) = s_c \; (1-p_b) \; \EX_{\theta}\left[\sum_{k=1}^{K(\theta)} u_{i}(a^*(\theta), \theta, k) \given[\Big] \theta_{i,0} = 0\right].
\]
We can write the expectation in the above formula as follows.
\begin{align*}
\EX_{\theta}&\left[\sum_{k=1}^{K(\theta)} u_{i}(a^*(\theta), \theta, k) \given[\Big] \theta_{i,0} = 0\right] \\
= & \sum_{\ell=1}^{\kmax} \EP(K(\theta) = \ell) \; \EX_{\theta}\left[\sum_{k=1}^{\ell} u_i(a^*(\theta), \theta, k) \given[\Big]  \theta_{i,0} = 0, K(\theta) = \ell\right] \\ 
= & \sum_{\ell=1}^{\kmax} \EP(K(\theta) = \ell) \sum_{k=1}^{\ell} \EX_{\theta}\left[u_i(a^*(\theta), \theta, k) \given \theta_{i,0} = 0, K(\theta) = \ell\right] \\ 
= & \sum_{\ell=1}^{\kmax} \sum_{k=1}^{\ell} \EP(K(\theta) = \ell) \; \EX_{\theta}\left[R_i(a^*(\theta), \theta, k) - C_i(a^*(\theta), \theta, k) \given \theta_{i,0} = 0, K(\theta) = \ell\right].
\end{align*}

The second equality holds because the expected value of the sum of random variables is equal to the sum of their expectations. 
Given \refequ{equ:our_reward}, we have the following for $R_i$.
\begin{align*}
\EP(K&(\theta) = \ell) \; \EX_{\theta}\left[R_i(a^*(\theta), \theta, k)  \given \theta_{i,0} = 0, K(\theta) = \ell\right] \\
= & \; (R^{b}(k) / n_{rs}) \sum_{j \in N_i} \EP(K(\theta) = \ell) \; \EP(\theta_{j,k} \le p_c, a^*_j(\theta) = C \given \theta_{i,0} = 0, K(\theta) = \ell) \\
& + R^c(k) \; \EP(K(\theta) = \ell) \; \EP( \theta_{i,k} \le p_c \given \theta_{i,0} = 0, K(\theta) = \ell) \\
= & \; (R^{b}(k) / n_{rs}) \sum_{j \in N_i} \EP(K(\theta) = \ell) \; \EP(\theta_{j,k} \le p_c, \theta_{j,0} = 0 \given K(\theta) = \ell) \\
& + R^c(k) \; \EP(K(\theta) = \ell) \; \EP( \theta_{i,k} \le p_c \given K(\theta) = \ell) \\
= & \; (R^{b}(k) / n_{rs}) \sum_{j \in N_i} \EP(\theta_{j,k} \le p_c, \theta_{j,k} = 0) \; \EP(K(\theta) = \ell \given \theta_{j,k} \le p_c, \theta_{j,0} = 0) \\
& + R^c(k) \; \EP(\theta_{i,k} \le p_c) \; \EP(K(\theta) = \ell \given \theta_{i,k} \le p_c) \\
= & \; (R^{b}(k) / n_{rs}) \sum_{j \in N_i} p_c \; (1 - p_b) \; \EP(K(\theta) = \ell \given \theta_{j,k} \le p_c, \theta_{j,0} = 0) \\
& + R^c(k) \; p_c \; \EP(K(\theta) = \ell \given \theta_{i,k} \le p_c).
\end{align*}
The second equality holds since there is no dependency between any two elements of $\theta$, and according to the definition of $a^*$, $a^*_j(\theta) = C$ if and only if $\theta_{j,0} = 0$.
The third equality follows simply from Bayes' rule.

The number of steps needed to complete the Algorand protocol depends on whether the highest-priority committee member at steps $k \in \{ 1,  7, 10, \dots, \kmax \}$ is corrupted (see Theorem 1 in \cite{CM2019}).
Let $E_j$ be a binary random variable that takes value 1 if agent $j$ is not the highest-priority committee member at any of steps $k \in \{ 1,  7, 10, \dots, \kmax \}$ and takes value 0 otherwise.
If $E_j = 1$, then agent $j$'s action and type do not affect the number of steps it takes the Algorand protocol to complete.

Agent $j$ becomes the highest-priority committee member at any step $k$ if $\theta_{j,k} < \theta_{j^\prime,k}$ for all $j^\prime \in N$.
$\theta_{j,k}$ is uniformly distributed between 0 and 1.
Therefore, the probability that an agent is the highest-priority committee member at any steps is $1/n$.
Since $\theta_{j,k}$'s are independent and identically distributed, we have $\EP(E_j = 0) \le \kmax/3n$, and we can write:
\begin{align}
\label{equ:asymptotic-independence}
\EP(K&(\theta) = \ell \given \theta_{j,k} \le p_c, \theta_{j,k} = 0) \nonumber \\
= & \EP(K(\theta) = \ell \given \theta_{j,k} \le p_c, \theta_{j,k} = 0, E_j = 0) \; \EP(E_j = 0) \nonumber \\
& + \EP(K(\theta) = \ell \given \theta_{j,k} \le p_c, \theta_{j,k} = 0, E_j = 1) \; \EP(E_j = 1) \nonumber \\
\simeq & \EP(K(\theta) = \ell \given \theta_{j,k} \le p_c, \theta_{j,k} = 0, E_j = 1) = \EP(K(\theta) = \ell \given E_j = 1) = \EP(K(\theta) = \ell).
\end{align}
The first equality follows from the law of total probability.
The second equality holds asymptotically since $\EP(E_j = 0)$ goes to zero as $n$ goes to infinity.
The third equality holds since given $E_j = 1$, the number of steps needed to complete the protocol does not depend on agent $j$'s type.
Finally, the fourth equality holds since $E_j = 1$ and $K(\theta) = \ell$ are independent events.
Knowing that agent $j$ is not the highest-priority agent at some steps does not reveal any information about the number of steps needed to complete the protocol.
The same argument can be applied to show that $\EP(K(\theta) = \ell \given \theta_{i,k} \le p_c) = \EP(K(\theta) = \ell)$.

Given \refequ{equ:cost}, we have  the following for $C_i$.
\begin{align}
\label{equ:expected-costs}
\EP(K&(\theta) = \ell) \; \EX_{\theta}\left[C_i(a^*(\theta), \theta, k)  \given \theta_{i,0} = 0, K(\theta) = \ell\right] \nonumber \\
= & \; \EP(K(\theta) = \ell) \; (\bar{C}^{b}(k) + C^c(k) \; \EP( \theta_{i,k} \le p_c \given \theta_{i,0} = 0, K(\theta) = \ell)) \nonumber \\
= & \; \EP(K(\theta) = \ell) \; \bar{C}^{b}(k) + C^c(k) \; \EP(\theta_{i,k} \le p_c) \; \EP(K(\theta) = \ell \given \theta_{i,k} \le p_c) \nonumber \\
= & \; \EP(K(\theta) = \ell) \; \bar{C}^{b}(k)  +  C^c(k) \; p_c \; \EP(K(\theta) = \ell).
\end{align}
The second equality follows from Bayes' rule, and the third equality holds as shown for $R_i$.
Note that for this equality equality, we also use the fact that there is no dependency between $\theta_{i,0}$ and $\theta_{i,k}$.
This completes the proof.

\section{Algorand}
\label{sec:algorand}

In this section we present a summary of the Algorand protocol \cite{GHM+2017, CM2019}%
\footnote{We only present the details that are needed for our game-theoretic model.
Interested readers could see \cite{GHM+2017, CM2019} for a detailed description and analysis of the protocol}.
The protocol maintains a public, permissionless blockchain.
Adding a new block to the blockchain requires multiple steps.
At each step, a committee of randomly selected nodes is formed.
Committee members play different roles at different steps.
Each committee member proposes a new block at step 1 (see \refsec{sec:block-proposal}).
At the subsequent steps, committee members vote to reach consensus on the block that should be added to the blockchain (see \refsecs{sec:reduction}{sec:bba}).

\subsection{Step 1 (Block Proposal)}
\label{sec:block-proposal}

\refalg{alg:block-proposal} summarizes the procedure executed by all nodes at step one.
First, nodes run the sortition algorithm to determine whether they are in the committee at step 1.
A node $i$ is in the committee if the corresponding $.x_{i,1}$ is less than or equal to $p_c$.
A node that is in the committee generates a candidate block%
\footnote{A candidate block contains a set of pending transactions that a node has heard about.}
and propagates the block to other nodes using a peer-to-peer gossip protocol%
\footnote{A gossip protocol is a communication process that disseminates data to all nodes in the system.
In the most common implementation, upon receiving a message, a node selects a small random set of peers to gossip the message to.
To prevent loops, nodes do not relay the same message twice.} (see \refsec{sec:algorand-prot}).

Nodes also propagate the committee credentials.
This allows other nodes to verify the sortition results of committee members.
The hashes of the credentials are used to prioritize committee members.
The smaller this hash is, the higher the priority of the node will be.
Prioritizing committee members enables nodes to reach consensus on a single block proposed by the highest-priority committee member.

\begin{algorithm}[!t]
  ($x_{i,1}$, $\sigma_{i,1}$) $\gets$ Sortition(1)\;
  \If{$.x_{i,1}$ $\le$  $p_c$}{
    $B$ $\gets$ Generate block proposal\;
    Propagate($B$, $\sigma_{i,1}$)\;
  }
\caption{Step-1 procedure}
\label{alg:block-proposal}
\end{algorithm}

\subsection{Step 2 \& 3 (Graded Consensus)}
\label{sec:reduction}

At step 2, nodes run the first step of the Graded Consensus (GC) protocol.
The main goal of the graded consensus (GC) protocol is to reduce the number of candidate blocks to one and to convert the problem of reaching consensus on an arbitrary value (the hash of a block) to reaching consensus on a single bit.
\refalg{alg:gc-protocol1} shows the pseudo-code for step 2.
Nodes wait a fixed time period to receive block proposals.
While waiting, nodes validate any received blocks and gossip valid ones to their neighbors according to the gossip protocol.
Once the time period expires, nodes invoke the sortition algorithm to check whether they are in the committee for this step.
If selected, the node generates its vote according to the first step of the GC protocol.
The node then propagates its vote.
The node also propagates its' credential to allow other nodes to verify its membership.
Note that if $M$ sub-nodes of a node are selected, then the node's vote is counted as $M$ \emph{sub-votes}.

\begin{algorithm}[!t]
  Validate and gossip received block proposals for a fixed time period\;
  ($x_{i,2}$, $\sigma_{i,2}$) $\gets$ Sortition(2)\;
  \If{$.x_{i,2}$ $\le$  $p_c$}{
    $v$ $\gets$ Generate vote according to GC's 1st step\;
    Propagate($v$, $\sigma_{i,2}$)\;
  }
\caption{Step-2 procedure}
\label{alg:gc-protocol1}
\end{algorithm}

\refalg{alg:gc-protocol2} summarizes the procedure for step 3.
In this step, nodes run the second step of the GC protocol.
First, nodes validate and gossip any received votes from step 2 while waiting a fixed period of time.
Nodes then run the sortition algorithm.
If a node is in the committee for this step, then it generates its vote according to the second step of the GC protocol.
The node then propagates its vote and credential.

\begin{algorithm}[!t]
  Validate and gossip received step-2 votes for a fixed time period\;
  ($x_{i,3}$, $\sigma_{i,3}$) $\gets$ Sortition(3)\;
  \If{$.x_{i,3}$ $\le$  $p_c$}{
    $v$ $\gets$ Generate vote according to GC's 2nd step\;
    Propagate($v$, $\sigma_{i,3}$)\;
  }
\caption{Step-3 procedure}
\label{alg:gc-protocol2}
\end{algorithm}

\subsection{Step 4 to \texorpdfstring{K\textsubscript{max}}{K\_max} (Binary Byzantine Agreement) }
\label{sec:bba}
At step 4, nodes compute the output of the GC protocol and start the first step of the Binary Byzantine Agreement (BBA*) protocol.
The main goal of the Binary Byzantine Agreement BBA* protocol is to reach consensus among all nodes on a binary value.
The pseudo-code for step 4 is shown in \refalg{alg:bba1}.
Nodes first validate and gossip the received votes from step 3 for a fixed time period.
After this period expires, each node computes the output of the GC protocol, $(v,b)$,  where $v$ is a string and $b$ is a single bit.
Nodes then start the first step of the BBA* protocol by invoking the sortition algorithm.
If a node is selected as a committee member, it propagates its GC output and credential.

\begin{algorithm}[!t]
  Validate and gossip received step-3 votes for a fixed time period\;
  $(v,b)$ $\gets$ Output of GC protocol\;
  ($x_{i,4}$, $\sigma_{i,4}$) $\gets$ Sortition(4)\;
  \If{$.x_{i,4}$ $\le$  $p_c$}{
    Propagate($(v,b)$, $\sigma_{i,4}$)\;
  }
\caption{Step-4 procedure}
\label{alg:bba1}
\end{algorithm}

Starting from step 5, until a termination condition is met, nodes iterate three steps of the BBA* protocol: (i) coin-fixed-to-$0$, (ii) coin-fixed-to-$1$, and (iii) coin-fixed-to-flip.
In all three steps, nodes validate and gossip votes from the previous step for a fixed time period.
In coin-fixed-to-$0$ steps, nodes terminate if they receive more than $T \cdot \tau$ votes for $b=0$, where $\frac{2}{3} < T < 1$ is a fraction of $\tau$ that defines Algorand's voting threshold.
Otherwise, they run the sortition algorithm to check whether they are in the committee.
Nodes in the committee then calculate their new $b$ and propagate it with the credential.
The procedure is similar in coin-fixed-to-$1$ steps except that the termination condition is on $b = 1$.
\refalg{alg:bba2} and \refalg{alg:bba3} show the procedure for these steps.

\begin{algorithm}[!t]
  Validate and gossip received step-($k$ - 1) votes for a fixed time period\;
  \If{more than $T \cdot \tau$ sub-votes are received for $b=0$}{
    Terminate\;
  }

  ($x_{i,k}$, $\sigma_{i,k}$) $\gets$ Sortition($k$)\;
  \If{$.x_{i,k}$ $\le$  $p_c$}{
    $v$ $\gets$ Get $v$ from step 4\;
    $b$ $\gets$ Generate binary according to BBA*'s coin-fixed-to-$0$ step\;
    Propagate($(v,b)$, $\sigma_{i,k}$)\;
  }
\caption{Step-$k$ procedure for $5\le k \le K_{\max}$ and $k \equiv 2 \pmod 3$}
\label{alg:bba2}
\end{algorithm}

\begin{algorithm}[!t]
  Validate and gossip received step-($k$ - 1) votes for a fixed time period\;
  \If{more than $T \cdot \tau$ sub-votes are received for $b=1$}{
    Terminate\;
  }

  ($x_{i,k}$, $\sigma_{i,k}$) $\gets$ Sortition($k$)\;
  \If{$.x_{i,k}$ $\le$  $p_c$}{
    $v$ $\gets$ Get $v$ from step 4\;
    $b$ $\gets$ Generate binary according to BBA*'s coin-fixed-to-$1$ step\;
    Propagate($(v,b)$, $\sigma_{i,k}$)\;
  }
\caption{Step-$k$ procedure for $6 \le k \le K_{\max}$ and $k \equiv 0 \pmod 3$}
\label{alg:bba3}
\end{algorithm}

If nodes do not terminate in coin-fixed-to-$0$ and coin-fixed-to-$1$ steps, they run the coin-fixed-to-flip step.
In this step, nodes that are selected in the committee calculate a new $b$ and propagate it.
\refalg{alg:bba4} summarizes the procedure for coin-fixed-to-flip steps.
These steps do not have a termination condition.
Therefore, after a coin-fixed-to-flip step, nodes start running the next coin-fixed-to-$0$ step.

\begin{algorithm}[!t]
  Validate and gossip received step-($k$ - 1) votes for a fixed time period\;
  ($x_{i,k}$, $\sigma_{i,k}$) $\gets$ Sortition($k$)\;
  \If{$.x_{i,k}$ $\le$  $p_c$}{
    $v$ $\gets$ Get $v$ from step 4\;
    $b$ $\gets$ Generate binary according to BBA*'s coin-fixed-to-flip step\;
    Propagate($(v,b)$, $\sigma_{i,k}$)\;
  }
\caption{Step-$k$ procedure for $7 \le k \le K_{\max}$ and $k \equiv 1 \pmod 3$}
\label{alg:bba4}
\end{algorithm}

\end{document}